\newcommand{\numColumns}{2}           
\newcommand{\isUseTableOfContent}{0}  
\newcommand{\isUseThanks}{1}          
\newcommand{\isUseThightHeaders}{1}   
\newcommand{\dateMode}{2}             
\newcommand{\myBibliographyStyle}{1}  
\newcommand{\myBibliographyFile}{support_files/references}

\newcommand{\myDocumentType}{Preprint}  
\newcommand{\myAuthor}{Marco Coraggio*, Mario di Bernardo}
\newcommand{\myTitle}{Data-driven design of complex network structures\\to promote synchronization}
\newcommand{\myHeaderAuthorTitle}{M. Coraggio, M. di Bernardo: Data-driven design of complex network structures to promote synchronization}
\newcommand{\myDate}{15 September 2023}
\newcommand{\myThanks}{%
This work was in part supported by the Research Project \textsc{``Sharespace''} funded by the European Union (EU HORIZON-CL4-2022-HUMAN-01-14.~SHARESPACE.~GA 101092889).

M.~Coraggio is with the Scuola Superiore Meridionale (SSM), School for Advanced Studies (marco.coraggio@unina.it).
M.~di Bernardo is with the Dept.~of Information Technology and Electrical Engineering, Univ.~of Naples Federico II, and with the SSM (mario.dibernardo@unina.it).}
\newcommand{\myAbstract}{%
We consider the problem of optimizing the interconnection graphs of complex networks to promote synchronization.
When traditional optimization methods are inapplicable, due to uncertain or unknown node dynamics, we propose a data-driven approach leveraging datasets of relevant examples. 
We analyze two case studies, with linear and nonlinear node dynamics.
First, we show how including node dynamics in the objective function makes the optimal graphs heterogeneous.
Then, we compare various design strategies, finding that the best either utilize data samples close to a specific Pareto front or a combination of a neural network and a genetic algorithm, with statistically better performance than the best examples in the datasets.}

\ifnum\numColumns=1
  \documentclass[onecolumn]{article}
\else\ifnum\numColumns=2
  \documentclass[twocolumn]{article}
\else
  \GenericError{}{Value of \protect\numColumns\space is not acceptable}{}{}
\fi\fi

\usepackage{newtxtext}        
\usepackage[utf8]{inputenc}   
\usepackage[T1]{fontenc}      
\usepackage[shortlabels]{enumitem}  
\usepackage[svgnames]{xcolor} 
\usepackage{lipsum}	          
\usepackage{soul}             

\usepackage{amsmath}    
\usepackage{amsfonts}   
\usepackage{amsthm}     
\usepackage{mathtools}  
\usepackage{newtxmath}  
\usepackage{bm}         

\usepackage[caption=false,font=footnotesize]{subfig}
\usepackage{graphicx}	          
\usepackage{booktabs}   	      
\usepackage[export]{adjustbox}  
\usepackage{mdframed}		        
\usepackage[ruled,vlined,linesnumbered]{algorithm2e}  

\usepackage{geometry}		    
\usepackage{fancyhdr}       
\usepackage{titlesec}       
\ifnum\isUseThightHeaders=1
	\usepackage{titling}      
\fi
\usepackage{abstract}		    
\usepackage[hidelinks]{hyperref}  
\usepackage{datetime}		    

\ifnum\myBibliographyStyle=0       
	\usepackage[numbers]{natbib}
	\newcommand{\bibliographyStyleName}{IEEEtran}
\else\ifnum\myBibliographyStyle=1  
	\usepackage[square]{natbib}
	\newcommand{\bibliographyStyleName}{IEEEtranSN}
	\renewcommand{\cite}{\citep} 
\else
	\GenericError{}{Value of \protect\myBibliographyStyle\space is not acceptable}{}{}
\fi\fi

\graphicspath{{figures/}}         
\numberwithin{equation}{section}	
\allowdisplaybreaks               

\makeatletter
\patchcmd{\@algocf@start}
  {-1.5em}
  {0pt}
  {}{}
\makeatother

\SetKwComment{tcc}{ {$\triangleright$} }{}  
\SetKwComment{tcp}{ {$\triangleright$} }{}  

\SetCommentSty{myCommentStyle}

\ifnum\numColumns=1
\else\ifnum\numColumns=2
	\let\oldtabular\tabular
	\renewcommand{\tabular}{\footnotesize\oldtabular}
\else
  \GenericError{}{Value of \protect\numColumns\space is not acceptable}{}{}
\fi\fi
\setlist{noitemsep,topsep=0.5ex}     
\definecolor{myLightBlue}{rgb}{0.09 0.36 0.63}
\colorlet{colorAccent}{myLightBlue}
\colorlet{colorAccentLight}{gray!15}
\hypersetup{  
  colorlinks,
  linkcolor = {DodgerBlue},
  citecolor = {DodgerBlue},
  urlcolor  = {DodgerBlue}
}

\newdateformat{myDateFormat}{\THEDAY \ \monthname[\THEMONTH] \THEYEAR}
\let\oldheadrule\headrule
\let\oldfootrule\footrule
\renewcommand{\headrule}{{\color{colorAccent}\oldheadrule}}
\renewcommand{\footrule}{{\color{colorAccent}\oldfootrule}}
\fancypagestyle{firstpage}{  
	\fancyhf{}
	\lhead{\sffamily \footnotesize \myDocumentType}
	\chead{}
	\ifnum\dateMode=0
		\rhead{\sffamily \footnotesize}
	\else\ifnum\dateMode=1
		\rhead{\sffamily \footnotesize \myDateFormat\today}
	\else\ifnum\dateMode=2
		\rhead{\sffamily \footnotesize \myDate}
	\else
		\GenericError{}{Value of \protect\dateMode\space is not acceptable}{}{}
	\fi\fi\fi
	\lfoot{}
	\cfoot{\sffamily\footnotesize \thepage}
  \rfoot{}

}

\titleformat{\part}[display]{\centering\sffamily\Huge\bfseries}{\textsc{\partname}\thepart\\ \vspace*{1ex}{\Huge\color{colorAccent}------ $\cdot$ ------}}{1.2ex}{\huge}
\titleformat{\chapter}{\sffamily\bfseries\Huge\color{colorAccent}}{\thechapter \	}{10pt}{}
\ifnum\isUseThightHeaders=1
	\titleformat*{\section}{\centering\large\sffamily\color{colorAccent}}
	\titleformat*{\subsection}{\sffamily\color{colorAccent}}
	\titlespacing\section{0pt}{3ex}{1ex}
	\titlespacing\subsection{0pt}{2.5ex}{0.5ex}
\else
	\titleformat*{\section}{\Large\sffamily\color{colorAccent}}
	\titleformat*{\subsection}{\large\sffamily\color{colorAccent}}
\fi
\titleformat*{\subsubsection}{\sffamily\itshape\color{colorAccent}\/}
\titleformat*{\paragraph}{\color{colorAccent} \itshape \sffamily}
\titlespacing*{\paragraph}{0pt}{0pt}{10pt}

\renewcommand{\abstractname}{}      

\newtheoremstyle{mythmstyle}
{0.7em}
{1em}
{\itshape}
{}
{\sffamily \color{colorAccent}}
{.}
{ }
{{\bfseries\thmname{#1}\thmnumber{ #2}}\thmnote{ (#3)}}  

\theoremstyle{mythmstyle}			 

\mdfdefinestyle{myshadedthm}{
	backgroundcolor  = colorAccentLight,
	linewidth        = 0pt,
	innerleftmargin  = 2ex,
	innerrightmargin = 2ex,
	innertopmargin   = 0ex
	}

\newmdtheoremenv[style=myshadedthm]{theorem}{Theorem}[section]
\newmdtheoremenv[style=myshadedthm]{lemma}[theorem]{Lemma}
\newmdtheoremenv[style=myshadedthm]{proposition}[theorem]{Proposition}
\newmdtheoremenv[style=myshadedthm]{corollary}[theorem]{Corollary}
\newmdtheoremenv[style=myshadedthm]{definition}[theorem]{Definition}
\newmdtheoremenv[style=myshadedthm]{assumption}[theorem]{Assumption}
\newmdtheoremenv[style=myshadedthm]{problem}[theorem]{Problem}
\newmdtheoremenv[style=myshadedthm]{remark}[theorem]{Remark}
\newmdtheoremenv[style=myshadedthm]{conjecture}[theorem]{Conjecture}

\newcommand{\BB}[1]{\mathbb{#1}}

\newcommand{\T}{^{\mathsf{T}}}
\newcommand{\R}[1]{\mathrm{#1}}
\newcommand{\B}[1]{\if#1\relax\bm{#1}\else\mathbf{#1}\fi} 
\newcommand{\C}[1]{\mathcal{#1}}
\newcommand{\abs}[1]{\left\lvert #1 \right\rvert}

\title{\color{colorAccent}{\Huge \sffamily \myTitle}}%
\ifnum\isUseThightHeaders=1
\fi
\author{\color{colorAccent}{\normalsize \sffamily \itshape \myAuthor\/}
\ifnum\isUseThanks=1  \thanks{\myThanks\/}  \fi}
\date{}
\renewcommand\footnotemark{}

\makeatletter
\ifx\@date\empty
\patchcmd{\@maketitle}{{\@bspredate \@date \@bspostdate} \maketitlehookd \par \vskip 1.5em}{\vskip 0.5em}{}{}
\fi
\makeatother%

\begin{document}

\maketitle
\thispagestyle{firstpage}  
\begin{abstract}
  \noindent \normalsize {\color{colorAccent}{\textsf{\textbf{Abstract. \;}}}}\myAbstract
\end{abstract}
\vspace{3ex}

\ifnum\isUseTableOfContent=1
  \tableofcontents
  \clearpage
\fi

\def\marginUnderSubfig{-4mm}

\section{Introduction}%
\label{sec:introducion}

In complex networks, the graph structure is a crucial component in determining the appearance of collective behavior such as synchronization, which is relevant in numerous applications, ranging from power systems to social networks, to biological processes	\cite{pikovskij2003synchronization}.
Thus, it is critical to devise tools to design network graphs that facilitate (or impede) 
synchronization.

In this paper, we introduce the data-driven network design problem, which serves as a flexible framework when traditional optimization methods are inapplicable (e.g., because knowledge of the node dynamics is incomplete or unavailable).
We explore two case studies, one with linear and the other with nonlinear node dynamics. 
The analysis shows that graph homogeneity, while important, is not enough to optimize synchronization-related metrics that include node dynamics.
Then, we present multiple data-driven network design strategies and assess their performance across different datasets. 
We find that the best strategies are those that generate suboptimal network structures by utilizing data samples close to a specific Pareto front or by leveraging the combination of a neural network and a genetic algorithm.

\paragraph{Related work}

The impact of a network's graph on synchronizability is typically measured by its eigenratio (the ratio between the largest and smallest non-zero eigenvalues of the associated Laplacian matrix) or its algebraic connectivity.
For high synchronizability, the former should be minimized \cite{pecora1998master}, while the latter should be maximized \cite{coraggio2018synchronization,coraggio2020distributed}.
Early studies showed that small-world networks have smaller eigenratios than random graphs \cite{barahona2002synchronization}, and that scale-free and small-world graphs become less synchronizable as they become more heterogeneous \cite{nishikawa2003heterogeneity}.
Crucially, in \cite{donetti2005entangled}, an iterative rewiring process revealed that graphs minimizing the eigenratio exhibit an \emph{entangled} structure.%
\footnote{An entangled graph has a homogeneous structure, characterized by low variance in degrees, in betweenness centralities, and in shortest path lengths, small diameter, large girth, large average of the shortest cycles from a vertex to itself, and an absence of community structure.}
In subsequent research \cite{donetti2006optimal}, it was observed that heterogeneity in coupling strength led to more heterogeneous optimally synchronizable graphs.

In \cite{nishikawa2006synchronization,fazlyab2017optimal,kempton2018selforganization}, 
optimal graphs were sought by assigning weights and/or directions to graphs' edges, or by assigning the frequencies of oscillator nodes.
In \cite{estrada2010design}, the authors introduced procedures to construct \emph{golden spectral networks}, which are sparse, highly synchronizable and robust to vertex/edge removal.
Recently, in \cite{lei2023new},  Lyapunov functions were used to design optimally synchronizable networks of oscillators, relying on knowledge of the nodes' frequencies.
Additional network design methods were surveyed in \cite{jalili2013enhancing}.

Notably, many previous studies assessing synchronization properties overlook the influence of node dynamics, despite evidence indicating its significance \cite{donetti2006optimal}.
An exception can be found in \cite{gorochowski2010evolving}, where a rewiring procedure demonstrated that optimal graphs may not necessarily exhibit an entangled structure, when node dynamics is considered.

To the best of our knowledge, data-driven approaches have not yet been employed as the primary tool for designing optimally synchronizable networks. 
Although, they have been used to control complex networks \cite{baggio2021datadriven,celi2023distributed} and to identify network graphs \cite{timme2007revealing}.

\section{Preliminaries}%
\label{sec:preliminaries}

\paragraph{Notation}

The $i$-th element of a vector $\B{x}$ is denoted by $x_i$;
$\R{re}(\cdot)$ is the real part;
$\R{round}(\cdot)$ is the nearest integer;
$\lceil \cdot \rceil$ and $\lfloor \cdot \rfloor$ are the nearest larger and smaller integers, respectively;
$\abs{\cdot}$ is the absolute value of a number or the cardinality of a set;
$\R{corr}(\cdot, \cdot)$ is the correlation;
$\R{tr}(\cdot)$ is the trace;
$\mu_2(\cdot)$ is the logarithmic $2$-norm;
$\lambda_i(\cdot)$ is the $i$-th eigenvalue (sorted from smallest to largest, when they are all real);
$k \text{-} \R{args} \max (\cdot)$ ($\min$) are the $k$ values that maximize (minimize) a quantity.

\paragraph{Graphs}

We always consider \emph{undirected} and \emph{unweighted} graphs \cite{boccaletti2006complex}.
Given a \emph{graph} $g = (\C{V}, \C{E})$, 
$\C{V}$ is the set of \emph{vertices} and
$\C{E}$ is the set of \emph{edges};
moreover, $n_\R{v} \coloneqq \abs{\C{V}}$ and $n_\R{e} \coloneqq \abs{\C{E}}$. 
We define $n_{\R{e}}^{\R{min}} \coloneqq n_\R{v} - 1$ and 
$n_{\R{e}}^{\R{max}} \coloneqq n_\R{v} (n_\R{v} - 1) / 2$.
$\B{L}(g)$ is the \emph{Laplacian matrix} of $g$.
The \emph{algebraic connectivity} of a connected graph is  $\lambda_2(\B{L})$ and its \emph{eigenratio} is $Q \coloneqq \lambda_{n_\R{v}}(\B{L}) / \lambda_2(\B{L})$.
The \emph{density} of a graph is $s \coloneqq \frac{2 n_{\R{e}}}{n_{\R{v}} (n_{\R{v}} - 1)}$.

\begin{definition}[Degree]\label{def:degree}
	The \emph{degree} $d_i$ of vertex $i$ is the number of edges connected to it.
	The \emph{mean degree} is $\R{mean}(\B{d}) \coloneqq \frac{1}{n_{\R{v}}}\sum_{i = 1}^{n_{\R{v}}} d_i$.
	The \emph{normalized degree deviation} of vertex $i$ is $\hat{d}_i \coloneqq \frac{d_i - \R{mean}(\B{d})}{n_{\R{v}}-1}$.
	The \emph{sample variance of degrees} is
	\(
		\R{var}(\B{d}) = \frac{1}{n_{\R{v}}-1}\sum_{i = 1}^{n_{\R{v}}} \left( d_i - \R{mean}(\B{d}) \right)^2
	\).
	The \emph{normalized variance of the node degrees} $\widehat{\R{var}}(\B{d})$ is $\frac{n_{\R{v}}-1}{n_{\R{v}} n_{\R{e}} (1-s)} \R{var}(\B{d})$ if $s \in \ ]0, 1[$ and is $0$ if $s = \{0, 1\}$ \cite{smith2020normalised}.
\end{definition}

We denote by $p_{jk}$ the number of \emph{shortest paths} from vertex $j$ to vertex $k$, and by $p_{jk}^i$ the number of these passing through vertex $i$ \cite{boccaletti2006complex}.

\begin{definition}[Betweenneess centrality]\label{def:betweenness}
		The \emph{betweenneess centrality} of vertex $i$ is
		$b_i \coloneqq \sum_{j, k \neq i} \frac{p_{jk}^i}{p_{jk}}$.
		The \emph{mean betweenness centrality} is 
		$\R{mean}(\B{b}) \coloneqq \frac{1}{n_{\R{v}}}\sum_{i = 1}^{n_{\R{v}}} b_i$.
		The \emph{normalized betweenness centrality deviation} is
		$\hat{b}_i \coloneqq \frac{b_i - \R{mean}(\B{b})}{(n_{\R{v}}-1)(n_{\R{v}}-2)/2}$.
		The \emph{sample variance of betweenness centralities} is 
		$\R{var}(\B{b}) = \frac{1}{n_{\R{v}}-1}\sum_{i = 1}^{n_{\R{v}}}(b_i - \R{mean}(\B{b}))^2$.
		The \emph{normalized variance of betweenness centralities} is%
		\footnote{Obtained by dividing $\R{var}(\B{b})$ by the sample variance of betweenness centralities of a star graph with infinite vertices, that is $\frac{1}{n_{\R{v}}-1} \big( \frac{(n_{\R{v}}-1)(n_{\R{v}}-2)}{2} \big) ^2$.}
		$\widehat{\R{var}}(\B{b}) = 4\frac{\R{var}(\B{b})}{(n_{\R{v}}-1)(n_{\R{v}}-2)^2}$.
\end{definition}

\section{Problem statement}%
\label{sec:problem_statement}

In general, we aim to find the graph structure of a complex network, which optimizes some objective function, in the presence of constraints.
We assume that lack of information or practical difficulties prevent the use of a traditional optimization algorithm, but that datasets of previous examples are available to inform the network design.

Formally, let $\C{S}$ be the set of continuous-time smooth dynamical systems, and let
$\C{S}^{n_\R{v}} \coloneqq \left( \C{S} \times \dots \times \C{S}\right)_{\text{${n_\R{v}}$ times}}$, for some ${n_\R{v}} \in \BB{N}_{\ge 2}$.
Let $\C{G}^{n_\R{v}}$ be the set of graphs with ${n_\R{v}}$ vertices.
Then, $\C{N}^{n_\R{v}} \coloneqq (\C{S}^{n_\R{v}} \times \C{G}^{n_\R{v}}, m)$ is the set of \emph{complex networks} with ${n_\R{v}}$ nodes, coupled through a coupling protocol $m$ (e.g., the linear diffusive one).
For example, $q \in \C{S}^{n_\R{v}}$ is a set of ${n_\R{v}}$ dynamical systems, $g \in \C{G}^{n_\R{v}}$ is a graph with ${n_\R{v}}$ vertices, and $\eta = (q, g, m)$ is a complex network with ${n_\R{v}}$ nodes.

Next, let
$J : \left( \bigcup_{{n_\R{v}} \in \BB{N}_{>0}} \C{N}^{n_\R{v}} \right) \rightarrow \BB{R}$ 
be the \emph{objective function}, measuring how good a network is with respect to some criterion, and let 
$\rho : \left( \bigcup_{{n_\R{v}} \in \BB{N}_{>0}} \C{G}^{n_\R{v}} \right) \rightarrow \BB{R}$ 
be the \emph{resource function}, measuring the resources consumed by a network.
Consider now a \emph{dataset} 
\begin{equation}\label{eq:dataset}
    \C{D} \coloneqq \left( \eta_h, J(\eta_h) \right)_{h \in \{1, \dots, n_{\R{d}}\}}    
\end{equation}
of $n_\R{d}$ data samples, each made of a complex network $\eta_h = (q_h, g_h, m) \in \C{N}^{n_\R{v}^h}$ with $n_\R{v}^h$ nodes and its associated objective value $J(\eta_h)$.
We aim to solve the following problem.

\begin{problem}[Data-driven network design]\label{prob:optimization_problem_general}
	Let $\eta^\diamond = (q^\diamond, g^\diamond, m^\diamond) \in \C{N}^{n_\R{v}^\diamond}$ be a complex network, where $q^\diamond \in \C{S}^{n_\R{v}^\diamond}$ is a set of $n_\R{v}^\diamond$ \emph{unknown} dynamical systems, $g^\diamond \in \C{G}^{n_\R{v}^\diamond}$ is a graph to be designed, with $n_\R{v}^\diamond$ vertices, and $m^\diamond$ is a fixed coupling protocol.
	Let $\C{D}$ be a dataset as in \eqref{eq:dataset}, with known network graphs $g_h$, known associated objective values $J(\eta_h)$, and unknown dynamical systems $q_h$, but with $q_h = q^\diamond \ \forall h$.
	Solve: $\max_{g^\diamond \in \C{G}^{n_\R{v}^\diamond}} J(\eta^\diamond)$ such that $\rho(g^\diamond) \le 0$.%
	\footnote{The problem can also be formulated with variations such as considering discrete-time dynamical systems, weighted graphs, equality constraints, etc.}
\end{problem}

Crucially, in Problem \ref{prob:optimization_problem_general}, it is impossible to compute $J(\eta^\diamond)$ for a given graph $g^\diamond$, as the node dynamics $q^\diamond$ are unknown.
The (possibly approximate) solution to Problem \ref{prob:optimization_problem_general} must be found exploiting the knowledge embedded in the dataset $\C{D}$.

\begin{remark}\label{rem:case_different_dynamics}
	A variation of Problem \ref{prob:optimization_problem_general} that is relevant for applications is that $q^\diamond$ and $q_h \ \forall h$ are allowed to be different but are known, although it is still impossible to optimize $J(\eta^\diamond)$ directly because either its  expression is unknown or the computation is unfeasible.
	In this case, different approaches from those presented in this paper should be employed; this matter will be the subject of future work.
\end{remark}

Next, we particularize the general Problem \ref{prob:optimization_problem_general} to two representative case studies.

\section{Case studies}%
\label{sec:case_studies}

\subsection{Case with linear node dynamics}%
\label{sec:linear_dynamics}%

We assume that the dynamical systems $q^\diamond$ are linear, scalar, heterogeneous, and stable, and that $m^\diamond$ is the linear diffusive coupling typically used in the literature \cite{scardovi2009synchronization}.
Hence, the dynamics of the complex network $\eta^\diamond$ are given by
\begin{equation}\label{eq:complex_network_linear}
	\dot{x}_i(t) = a_i x_i(t) + \sum\nolimits_{j = 1}^{n_{\R{v}}} L_{ij} (x_j(t) - x_i(t)), \ \  \forall i \in \{1, \dots, n_{\R{v}}\},
\end{equation}
where $x_i(t) \in \BB{R}$ is the state of dynamical system $i$, 
$a_i \in \BB{R}_{< 0}$,
and $[L_{ij}] = \B{L}(g^\diamond)$.
We let 
$\B{A} \coloneqq \R{diag}(a_1, \dots, a_{n_{\R{v}}})$ 
to rewrite \eqref{eq:complex_network_linear} as $\dot{\B{x}}(t) = (\B{A} - \B{L}) \B{x}$.
Note that \eqref{eq:complex_network_linear} always synchronizes to $\B{0}$, after a settling time determined by the spectrum of $\B{A}-\B{L}$.

We assume the goal is to find the network structure that minimizes the transient time to synchronization. 
Thus, we choose an objective function $J$ proportional to the dominant eigenvalue of network \eqref{eq:complex_network_linear}.
To give an expression for $J$, consider the slowest natural modes of the uncoupled and coupled dynamical systems, that is
$\lambda^\R{u} \coloneqq \max_i \R{re}(\lambda_i(\B{A}))$,
and
$\lambda^\R{c} \coloneqq \max_i \R{re}(\lambda_i(\B{A} - \B{L}))$, respectively.
The relation between $\lambda^\R{u}$ and $\lambda^\R{c}$ is given in the next Lemma, proved in the Appendix.

\begin{lemma}\label{lem:bounds_eigenvalue_coupled_net}
	It holds that $\beta \le \lambda^\R{c} \le \lambda^\R{u} < 0$, where
	$\beta \coloneqq \frac{1}{n_{\R{v}}} \left(\sum_{i = 1}^{n_{\R{v}}} a_i - (n_{\R{v}} (n_{\R{v}}-1))\right)$.
\end{lemma}

We take $J = - \frac{\lambda^\R{c} - \lambda^\R{u}}	{\abs{\beta - \lambda^\R{u}}}$.
From Lemma \ref{lem:bounds_eigenvalue_coupled_net}, $J \in [0, 1]$, with $0$ corresponding to no improvement over the uncoupled systems, and $1$ corresponding to the ideal fastest synchronization time.


Finally, we constrain the designed graph $g^\diamond$ to be connected ($\lambda_2(\B{L}) > 0$) and to have at most $n_\R{e}^*$ edges ($\R{tr}(\B{L}) \le \frac{n_\R{e}^*}{2}$).

\subsection{Case with nonlinear node dynamics}
\label{sec:nonlinear_dynamics}

In this case, we still assume a linear diffusive coupling $m^\diamond$, but the dynamics of the complex network $\eta^\diamond$ are given by 
\begin{equation}\label{eq:complex_network_nonlinear}
	\dot{x}_i = f(x_i) + \sum\nolimits_{j = 1}^{n_{\R{v}}} L_{ij} (x_j - x_i), \ \  \forall i \in \{1, \dots, n_{\R{v}}\},
\end{equation}
where $f : \BB{R} \rightarrow \BB{R}$ is a nonlinear dynamics; here, we set $f(x_i) = a_i (x_i - x_i^3)$, with $a_i \in \BB{R}_{> 0} \ \forall i$.
Network \eqref{eq:complex_network_nonlinear} may or may not synchronize, depending on the network structure $g^\diamond$ (i.e., on the properties of its associated graph Laplacian $\B{L}$).

\begin{remark}
	In Problem \ref{prob:optimization_problem_general}, the assumption that $q = q_h^\R{d} \ \forall h$ requires the initial conditions of the systems to be the same across the dataset, which can be restrictive in applications.
	Relaxing this assumption leads to the case described in Remark \ref{rem:case_different_dynamics}.
\end{remark}

Again, we assume $J$ is a metric related to the time required to achieve synchronization.
Namely, define the \emph{average state} $\tilde{x}(t) \coloneqq \sum_{i=1}^{n_\R{v}} x_i(t)$,
the \emph{node error} $e_i(t) \coloneqq x_i(t) - \tilde{x}(t)$,
and the \emph{total error} $e_\R{tot}(t) \coloneqq \sum_{i=1}^{n_\R{v}} e_i(t)$.
Then, let $t_{\R{sync}} $ be the smallest time instant such that $e_\R{tot}(t) \le e_{\R{thres}}, \forall t \in [t_{\R{sync}}, t_{\R{max}}]$, where $e_{\R{thres}}, t_{\R{max}} \in \BB{R}_{\ge 0}$; if such time instant does not exist, we take $t_{\R{sync}} = t_{\R{max}}$.
Then, we take $J = 1 - \frac{t_{\R{sync}}} {t_{\R{max}}} \in [0 , 1]$, so that $J = 0$ corresponds to synchronization being achieved at time $t_{\R{max}}$ or to no synchronization, whereas $J = 1$ corresponds to synchronization being reached at time $t = 0$.

As in the previous case study, we require the designed graph $g^\diamond$ to be connected and to have at most $n_\R{e}^*$ edges.

\subsection{Datasets}

To examine the case studies and illustrate our data-driven approach to network design, we consider five datasets: 
$\C{D}_{\R{middle}}^\ell$, 
$\C{D}_{\R{small}}^\ell$, 
$\C{D}_{\R{large}}^\ell$ for the linear case, and
$\C{D}_{\R{middle}}^{\R{n}\ell}$, 
$\C{D}_{\R{large}}^{\R{n}\ell}$ for the nonlinear case.
The subscripts refer to the size of the dataset (see Table \ref{tab:datasets}).
Each datasets is generated pseudo-randomly in $20$ \emph{iterations}, as described in the Appendix.

\begin{table}[t]
  \centering
  \begin{tabular}{lllll} 
      \toprule
      Dataset & $n_\R{v}$ & $n_\R{e}^*$ & mean num.~samples & coverage of decision space\\
      \midrule
      $\C{D}_{\R{middle}}^\ell$ & 20 & 45 & 641.85 & $4.090 \cdot 10^{-53} \%$\\
			$\C{D}_{\R{small}}^\ell$ & 10 & 20 & 259.50 & $1.883 \cdot 10^{-53} \%$\\
			$\C{D}_{\R{large}}^\ell$ & 20 & 45 & 4689.80 & $1.360 \cdot 10^{-8} \%$\\
			$\C{D}_{\R{middle}}^{\R{n}\ell}$ & 10 & 20 & 223.05 & $6.466 \cdot 10^{-10} \%$\\
			$\C{D}_{\R{large}}^{\R{n}\ell}$ & 10 & 20 & 4689.80 & $1.360 \cdot 10^{-8} \%$\\
      \bottomrule
  \end{tabular}
	\caption{Information on datasets (averaged over iterations).}
  \label{tab:datasets}
\end{table}

\section{Analysis of the datasets}%
\label{sec:analysis}

In Figure \ref{fig:datasets}, we portray the number of edges of the graphs in two representative datasets, $\C{D}_{\R{middle}}^\ell$ and $\C{D}_{\R{middle}}^{\R{n}\ell}$, together with the associated value of the objective function $J$, also stored in the datasets. 
As expected, the objective $J$ is found to increase nonlinearly for higher numbers of edges $n_\R{e}$.

\begin{figure}[t]
	\centering
	\subfloat[]{\includegraphics[max width=\columnwidth]{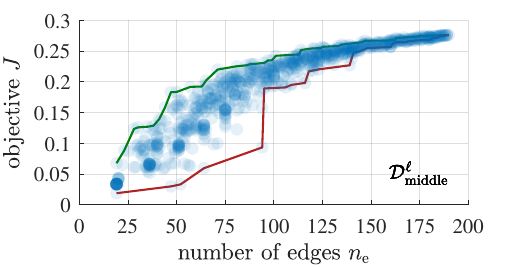}
	\label{fig:dataset_01_l}}\\
	\vspace*{\marginUnderSubfig}%
	\subfloat[]{\includegraphics[max width=\columnwidth]{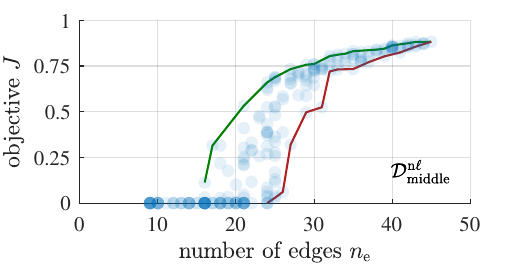}
	\label{fig:dataset_01_nl}}
	\caption{%
		Single iterations of representative datasets.
		The green and red lines are the good ($P_\R{g}$) and bad ($P_\R{b}$) Pareto fronts, respectively (§ \ref{sec:direct_strategies}).}
	\label{fig:datasets}
\end{figure}

Next, to assess whether having entangled graphs is important to maximize $J$, as it is to minimize the eigenratio $Q$ \cite{donetti2005entangled} (recall that $Q$ does not account for node dynamics), in Table \ref{tab:correlations_variances_J} we report $\R{corr}(\widehat{\R{var}}(\B{d}), J)$ and $\R{corr}(\widehat{\R{var}}(\B{b}), J)$,
comparing them with $\R{corr}(\widehat{\R{var}}(\B{d}), -Q)$ and $\R{corr}(\widehat{\R{var}}(\B{b}), -Q)$, respectively.
We see that having a large variance in $\B{d}$ and $\B{b}$ is detrimental both for $J$ and $Q$.
Surprisingly, $\abs{\R{corr}(\widehat{\R{var}}(\B{d}), J)} > \abs{\R{corr}(\widehat{\R{var}}(\B{d}), -Q)}$, suggesting an even larger effect of the entangled nature of the graphs on $J$ with respect to $Q$.

\begin{table}[t]
  \centering
  \begin{tabular}{ll|ll} 
      \toprule
      \multicolumn{4}{c}{\textit{Linear case study (§ \ref{sec:linear_dynamics}) --- Dataset $\C{D}_{\R{middle}}^\ell$}}\\
      \noalign{\smallskip}
      $\R{corr}(\widehat{\R{var}}(\B{d}), J)$ & $-0.430$ &
			$\R{corr}(\widehat{\R{var}}(\B{d}), -Q)$ &	$-0.181$\\
			$\R{corr}(\widehat{\R{var}}(\B{b}), J)$ & $-0.590$ &
			$\R{corr}(\widehat{\R{var}}(\B{b}), -Q)$ & $-0.559$\\
      \midrule
      \multicolumn{4}{c}{\textit{Nonlinear case study (§ \ref{sec:nonlinear_dynamics}) --- Dataset $\C{D}_{\R{middle}}^{\R{n}\ell}$}}\\
      \noalign{\smallskip}
      $\R{corr}(\widehat{\R{var}}(\B{d}), J)$ & $-0.260$ &
			$\R{corr}(\widehat{\R{var}}(\B{d}), -Q)$ &	$-0.088$\\
			$\R{corr}(\widehat{\R{var}}(\B{b}), J)$ & $-0.505$ &
			$\R{corr}(\widehat{\R{var}}(\B{b}), -Q)$ & $-0.578$\\
      \bottomrule
  \end{tabular}
	\caption{Metrics concerning entangled graphs (averaged over datasets iterations) when maximizing $J$ and $-Q$.}
  \label{tab:correlations_variances_J}
\end{table}

However, the entangled structure of a graph is not sufficient to optimize $J$.
To see this, in Figure \ref{fig:optimal_graphs}, we report the estimated optimal graphs $\B{L}^*$ obtained by optimizing $J$ and $Q$ through a genetic algorithm (whose parameters are in the Appendix), assuming node dynamics were known.
Indeed, including node dynamics in the objective function makes the optimal graph heterogeneous and its node degrees unequally distributed.

\begin{figure}[t]
	\centering
	\subfloat[]{%
	\includegraphics[scale=0.82]{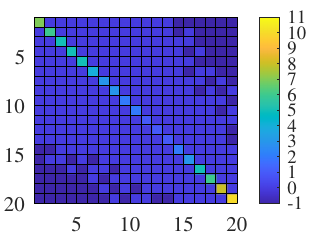}%
	\label{fig:optimal_graphs_dataset_01_dynamics}}%
	\hfill
	\subfloat[]{%
	\includegraphics[scale=0.82]{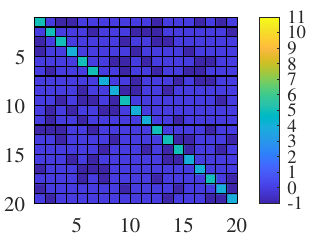}%
	\label{fig:optimal_graphs_dataset_01_eigenratio}}\\
	\vspace*{\marginUnderSubfig}%
	\subfloat[]{%
	\includegraphics[scale=0.82]{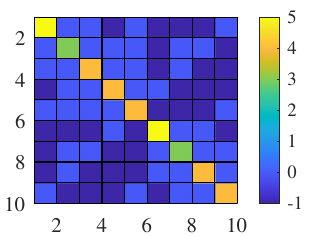}%
	\label{fig:optimal_graphs_dataset_04_dynamics}}%
	\hfill
	\subfloat[]{%
	\includegraphics[scale=0.82]{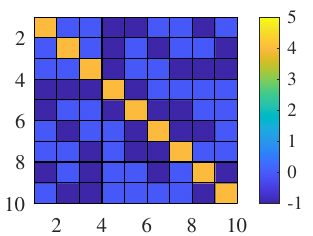}%
	\label{fig:optimal_graphs_dataset_04_eigenratio}}
	\caption{Estimated optimal Laplacian matrices $\B{L}^*$.
	(a, b) linear case in § \ref{sec:linear_dynamics};
	(c, d) nonlinear case in § \ref{sec:nonlinear_dynamics};
	(a, c) maximizing $J$;
	(b, d) maximizing $-Q$. 
	$\widehat{\R{var}}(d)$:
	(a) 0.232, (b) 0.007, (c) 0.036, (d) 0; 
	$\widehat{\R{var}}(b)$:
	(a) 0.148, (b) 0.002, (c) 0.013, (d) 0.001; 
	variance in lengths of the shortest paths: 
	(a) 0.602, (b) 0.319, (c) 0.373, (d) 0.292; 
	diameter:
	(a) 4, (b) 3, (c) 3, (d) 3; 
	girth:
	(a) 4, (b) 3, (c) 3, (d) 3; 
	average of the shortest cycles from a vertex to itself:
	(a) 4, (b) 3.5, (c) 3.200, (d) 3.}
	\label{fig:optimal_graphs}
\end{figure}

This fact is confirmed by Figure \ref{fig:correlations}, where we report $\R{corr}(\hat{d}_i, J)$ and $\R{corr}(\hat{b}_i, J)$ computed over the graphs in the datasets.
The results suggest that, in the linear case, slow nodes (i.e., with small $\abs{a}_i$) should have the largest degrees, which is in agreement with Figure \ref{fig:optimal_graphs_dataset_01_dynamics}.
In the nonlinear case, 
correlations are weaker, and we cannot draw definitive conclusions.

\begin{figure}[t]
	\centering
	\subfloat[]{\includegraphics[max width=\columnwidth]{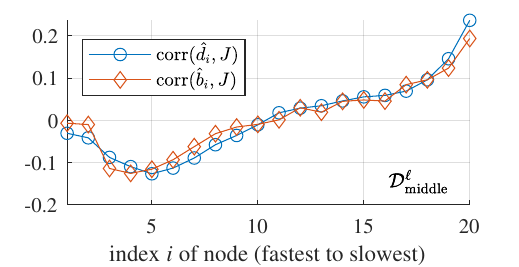}
	\label{fig:correlations_degrees_and_betw_dataset_linear}}\\
	\vspace*{\marginUnderSubfig}%
	\subfloat[]{\includegraphics[max width=\columnwidth]{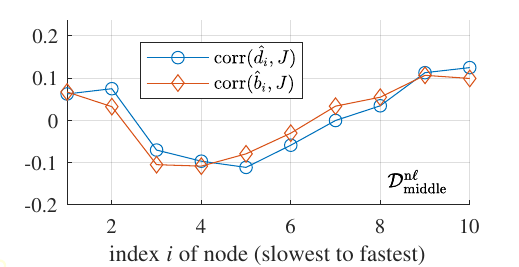}
	\label{fig:correlations_degrees_and_betw_dataset_nl}}
	\caption{Relevant correlations in representative datasets, averaged over datasets iterations. 
	$\R{corr}(i, \R{corr}(\hat{d}_i, J))$:
	(a) 0.837, (b) 0.303; 
	$\R{corr}(i, \R{corr}(\hat{b}_i, J))$:
	(a) 0.780, (b) 0.382.}
	\label{fig:correlations}
\end{figure}


\section{Data-driven network design strategies}%
\label{sec:algorithms}

Next, we present different approaches for data-driven network design.
These can be divided into \emph{indirect} strategies---where we first assess what are the features of the graphs in the dataset associated with the maximal $J$ and then generate a new network structure having those features---and \emph{direct} strategies---where the (sub)optimal graph is directly generated by appropriately combining those in the dataset.
Below, we let $\tilde{J}(\B{L})$ denote $J(q, g, m)$, as $q$ and $m$ are fixed.

\subsection{Indirect strategies}%

We start by describing the indirect strategies, based on extrapolating meaningful features from the graphs in the dataset.

\paragraph{Desired degree distribution (DDD)}

Let us define $\rho_i \coloneqq \R{corr}(\hat{d}_i, J)$ (cf.~Figure \ref{fig:correlations}).
As $\rho_i$ is a measure of how beneficial it is that node $i$ has a large degree for having a large $J$, we design the graph structure so as to have a degree distribution that replicates the shape of $\B{\rho}$.
To do so, we convert $\B{\rho}$ to a graphical%
\footnote{A vector $\B{d} \in \BB{N}_{> 0}^{n_\R{v}}$ is \emph{graphical} if there exists a graph (called a \emph{realization} of $\B{d}$) without loops or repeated edges in which vertex $i$ has degree $d_i$.
Graphicality can be checked, e.g., with the Erd\"os-Gallai condition \cite{blitzstein2011sequential}.}
degree distribution $\B{d}$ using Algorithm \ref{alg:degrees_from_curve}.
The procedure allocates degrees iteratively, subtracting each time a fixed quantity from the elements in $\B{\rho}$.
Then, to generate a graph with degree distribution $\B{d}$, we use a slightly modified version of \cite[Algorithm 1]{blitzstein2011sequential}, which allocates edges iteratively, each time subtracting $1$ from the elements of $\B{d}$.%
\footnote{Our modification is that, when allocating edges, the end vertex is chosen as the vertex with the larger value of $d_i$, rather than a pseudorandom one.}

\begin{algorithm}[t]
	\caption{Degrees distribution from vector}
	\label{alg:degrees_from_curve}
	\KwIn{Vector $\B{\rho} \in \BB{R}^{n_\R{v}}$; 
	num.~edges $n_{\R{e}}$ ($n_{\R{e}}^{\R{min}} \le n_{\R{e}} \le n_{\R{e}}^{\R{max}}$).}
	\KwOut{Degree distribution $\B{d} \in \BB{N}_{\ge 1}^{n_\R{v}}$.}

	$\B{\rho} \leftarrow \B{\rho} - \min_i \rho_i$\;
	$\B{d} \leftarrow \B{1}$%
	\tcp*{ensure connectedness}
	$n_{\B{d}} \leftarrow 2 n_{\R{e}} - n_{\R{v}}$%
	\tcp*{number of degrees to assign}
	$\Delta = \sum_i \rho_i / n_\B{d}$%
	\tcp*{a decrement unit}
	\While{$n_{\B{d}} > 0$}{
		$i \leftarrow \arg \max_j \rho_j$;\quad
		$d_i \leftarrow d_i + 1$;\quad
		$n_{\B{d}} \leftarrow n_{\B{d}} - 1$\;
		\leIf{$d_i < n_\R{v} - 1$}%
		{$\rho_i \leftarrow \rho_i - \Delta$;}
		{$\rho_i \leftarrow -\infty$}
	}
	\While{$\B{d}$ is not graphical}{
		$i \leftarrow \arg \max_j d_j$;\quad
		$d_{i} \leftarrow d_i - 1$\;
	}
\end{algorithm}










\paragraph{Neural network and genetic algorithm (NNGA)}

We consider a neural network (NN) that outputs an approximation of $\tilde{J}(\B{L})$, and takes as input: 
the off-diagonal elements of $\B{L}$,
$\lambda_2(\B{L})$, 
$\lambda_{n_\R{v}}(\B{L})$,
$n_\R{e}$,
$\hat{d}_i \ \forall i$, 
$\widehat{\R{var}}(\B{d})$, 
global and local clustering coefficients, 
average and variance of the shortest paths, 
the diameter, and
eigenvector centralities.
The NN is trained on the pairs of graphs and associated objective values $J$ in the dataset.
As the NN approximates $\tilde{J}(g)$, it is then used to run a numerical optimization through a genetic algorithm, to seek the optimal graph.
All parameters are in the Appendix.

\subsection{Direct strategies}%
\label{sec:direct_strategies}

We denote the Laplacian matrices of the graphs in the dataset by 
$\C{L} \coloneqq \{ \B{L}_1, \B{L}_2, \dots, \B{L}_{n_\R{d}} \}$.
Algorithm \ref{alg:combination_laplacian} generates a new graph structure by combining a subset of those in the dataset, say $\C{L}^\R{c} \subseteq \C{L}$, according to some weights $w_1, w_2, \dots$.
A (sub)optimal graph that attempts to maximize $J$ can then be obtained by careful selection of $\C{L}^\R{c}$ and the associated weights.
In the following, in Algorithm \ref{alg:combination_laplacian} we always take $n_\R{e}^\R{out} = n_\R{e}^*$; moreover, we define the set of graphs with $e$ edges as
${\C{L}}_e \coloneqq \{ \B{L} \in \C{L} \mid n_\R{e}(\B{L}) = e \}$, for $e \in \{n_\R{e}^{\R{min}}, \dots, n_\R{e}^{\R{max}}\}$, and the mean objective of such graphs as
$B_e \coloneqq \R{mean}_{\B{L}_i \in {\C{L}}_e} \tilde{J}(\B{L}_i)$.

Next, we propose a set of strategies to select the graphs from the dataset to be combined and the associated weights.

\begin{algorithm}[t]
	\caption{Combination of undir.~unweighted graphs}
	\label{alg:combination_laplacian}

	\KwIn{$n_\B{L}$ Laplacian matrices $ \B{L}_1, \B{L}_2, \dots, \B{L}_{n_\B{L}}$;
	weights $w_1, w_2, \dots, w_{n_\B{L}}$;
	num.~edges in output graph $n_\R{e}^{\R{out}}$.}
	\KwOut{Combined Laplacian matrix $\B{L}^\R{out}$.}

	\For{all possible edges $\{j, k\}$}{
		$z_{jk} \gets \sum_{i = 1}^{n_\B{L}} w_i \cdot  [-\B{L}_i]_{jk}$\;
	}
	$\C{E}_{\R{selected}} \leftarrow n_\R{e}^{\R{out}}\text{-}\R{args} \max_{\{j, k\}} z_{jk}$\;%
	Build a Laplacian matrix $\B{L}^\R{out}$ with edges $\C{E}_{\R{selected}}$.
\end{algorithm}

\paragraph{All graphs (A)}

We combine all graphs in the dataset, i.e., $\C{L}^\R{c} = \C{L}$, associating to graph $i$ a weight $w_i = \tilde{J}(\B{L}_i)^\alpha$; in particular, we select $\alpha = 3$.

\paragraph{All graphs normalized (AN)}

We again combine all graphs in the dataset ($\C{L}^\R{c} = \C{L}$), but with weights selected as $w_i = ( \tilde{J}(\B{L}_i) - B_{n_\R{e}(\B{L}_i)} )^\alpha$, choosing $\alpha = 3$.

\paragraph{Best and worst graphs for every fixed number of edges (BWNE)}

$\C{L}^\R{c}$ contains the fraction $p = 0.1$ of the best and worst graphs in the sets $\C{L}_e$ for each fixed  number of edges $e$.
More formally, $\C{L}^\R{c} = \C{L}_\R{best} \cup \C{L}_\R{worst}$, where, letting 
$p_e \coloneqq \min \{1, \ \R{round} \left( p \abs{\C{L}_e} \right) \}, \forall e \in  \{n_\R{e}^{\R{min}}, \dots, n_\R{e}^{\R{max}}\}$, we let
\begin{align*}
	\C{L}_\R{best} &\coloneqq \bigcup_{e \in \{n_\R{e}^{\R{min}}, \dots, n_\R{e}^{\R{max}} \}} p_e \text{-} \R{args} \max_{\B{L} \in \C{L}_e} \tilde{J}(\B{L}),\\
	\C{L}_\R{worst} &\coloneqq \bigcup_{e \in \{n_\R{e}^{\R{min}}, \dots, n_\R{e}^{\R{max}} \}} p_e \text{-} \R{args} \min_{\B{L} \in \C{L}_e} \tilde{J}(\B{L}).
\end{align*}
We take $w_i = +1$ if $\B{L}_i \in \C{L}_\R{best}$ and $w_i = -1$ if $\B{L}_i \in \C{L}_\R{worst}$.

\paragraph{Pareto front (PF)}

Let $\C{L}_{P_\R{g}}$ be the set of graphs in the dataset that are \emph{Pareto optimal} with respect to having small $n_\R{e}$ and being associated to a large $J$ \cite{censor1977pareto}, and are associated to $J > 0.01$.
Then, we compute the ``good'' Pareto front $P_\R{g} : \BB{N} \rightarrow \BB{R}$ (which associates to some number of edges $n_{\R{e}}$ a value of $J$; depicted as a green line in Figure \ref{fig:datasets}) by linearly interpolating the points in the $n_{\R{e}}$-$J$ plane associated to the graphs in $\C{L}_{P_\R{g}}$.
%
%
%
%
%
Next, we define the normalized distance of a graph in the dataset with Laplacian $\B{L}$ from the Pareto front as
\begin{equation}\label{eq:ranking_pareto}
	\delta_{P_\R{g}}(\B{L}) \coloneqq \begin{dcases}
		\frac{
			P_\R{g}(n_\R{e}(\B{L})) - \tilde{J}(\B{L})
		}{
			P_\R{g}(n_\R{e}(\B{L})) - B_{n_\R{e}(\B{L})}
		}, &\parbox[m]{.4\columnwidth}{
			if $P_\R{g}$ is defined for $n_\R{e}(\B{L})$ and $P_\R{g}(n_\R{e}(\B{L})) \ne B_{n_\R{e}(\B{L})}$,}\\
		\infty, &\text{otherwise}.
	\end{dcases}		
\end{equation}
Then, letting $p = 0.04$, we set $\C{L}^\R{c} = k \text{-} \R{args} \min_{\B{L}} \delta_{P_\R{g}}(\B{L})$, where $k = \max\{ |\C{L}_{P_\R{g}}|, \left\lceil{p \abs{\C{L}}}\right\rceil\}$.
The weights are $w_i = e^{-\delta_{P_\R{g}}(\B{L}_i)}$.

\paragraph{Double Pareto front (DPF)}

In this strategy, the set $\C{L}^\R{c}$ of graphs to be combined contains (i) those used in the ``Pareto front'' strategy, with the same weights, and (ii) a portion of the graphs closest to the ``bad'' Pareto front $P_\R{b}$ (depicted as a red line in Figure \ref{fig:datasets}), found interpolating Pareto optimal graphs with a maximal number of edges $n_\R{e}$ and a minimal value of $J$ (in this case, we do not exclude graphs associated to $J \le 0.01$).
The selection procedure remains the same,
except that the weights are chosen as $w_i = - e^{-\delta_{P_\R{b}}(\B{L}_i)}$, where $\delta_{P_\R{b}}$ is the normalized distance from $P_\R{b}$, computed using an expression analogous to \eqref{eq:ranking_pareto}.

\def\maxWidthPlotStrategies{0.95}
\begin{figure}[!ht]
	\centering
	\includegraphics[trim={0 2mm 0 1mm},clip,max width=\maxWidthPlotStrategies\columnwidth]{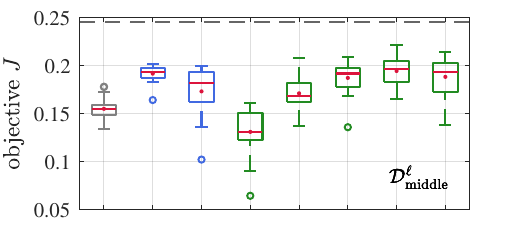}\\
	\includegraphics[trim={0 2mm 0 1mm},clip,max width=\maxWidthPlotStrategies\columnwidth]{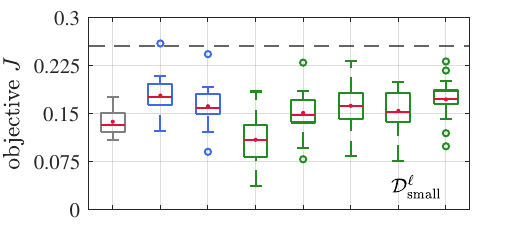}\\
	\includegraphics[trim={0 2mm 0 1mm},clip,max width=\maxWidthPlotStrategies\columnwidth]{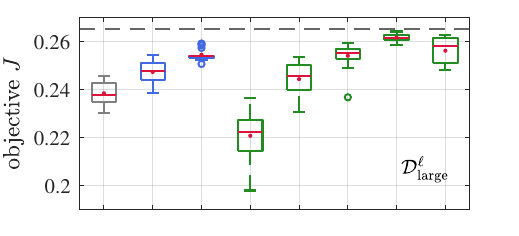}\\
	\includegraphics[trim={0 2mm 0 1mm},clip,max width=\maxWidthPlotStrategies\columnwidth]{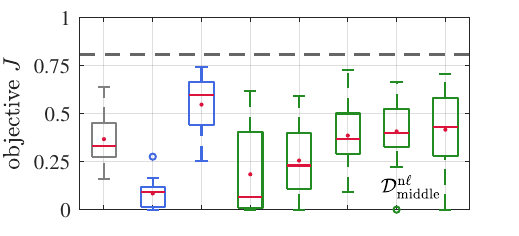}\\
	\includegraphics[trim={0 0mm 0 1mm},clip,max width=\maxWidthPlotStrategies\columnwidth]{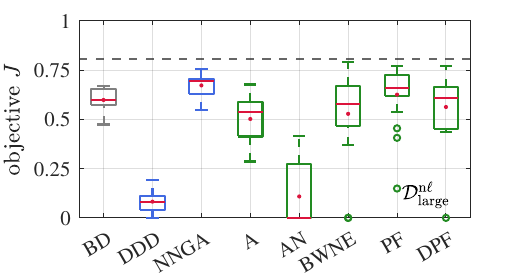}
	\caption{Box plots of the values of $J$ obtained by the strategies in § \ref{sec:algorithms}, in datasets' iterations.
	``BD'' is the best data sample; red dots are means; red lines are medians; circles are outliers; dashed grey lines are $J^*$.}
	\label{fig:strategies}
\end{figure}

\section{Validation and discussion}%
\label{sec:numerical_results}

In Figure \ref{fig:strategies}, we report the values of  $J$ obtained by the graphs designed through the strategies presented in Section \ref{sec:algorithms} over all iterations of all datasets, and compare them with the best values of $J$ found in the datasets.
%
%
We also report the value $J^*$ of the objective associated to the optimal graphs $\B{L}^*$ found with a genetic algorithm, assuming the dynamics were known.

We find that the DDD strategy performs well in the linear case (panels 1, 2, 3), but worse in the nonlinear one (panels 4, 5). 
The results suggest that the degree distribution has a greater effect on $J$ in the linear case, which is in agreement with the results reported in Figure \ref{fig:correlations}.
On the other hand, the NNGA strategy performs better than the best data sample in all cases, even when the dataset is small (panel 2).
%
The A and AN strategies performed the worst (even if $\alpha$ is changed).
This demonstrates that, counterintuitively, taking information from all graphs in the dataset can be detrimental; we believe this might be because the (important) information contained in the graphs closest to the good and bad Pareto fronts becomes obfuscated when $\C{L}^\R{c}$ is too large.
Indeed, the BWNE and DPF strategies have a smaller $\C{L}^\R{c}$ and yield better $J$.
The PF strategy, having smaller $\C{L}^\R{c}$, performs even better, with mean and quartiles always higher than the best data samples, and close to $J^*$ in relatively large datasets (panels 3, 5).

In conclusion, the NNGA and PF strategies proved to be the best ones, although the former requires much longer computation times than the latter.
As expected, most strategies tend to perform better in the linear case study than in the nonlinear one, and yield better results when the dataset is larger (compare panels 1 to 3 and 4 to 5).




\section*{Appendix}

\begin{proof}[Proof of Lemma \ref{lem:bounds_eigenvalue_coupled_net}]
We first prove that $\lambda^\R{c} \le \lambda^\R{u}$.
As $\B{A}$ and $\B{L}$ are symmetric, we have $\mu_2(\B{A} - \B{L}) = \max_i \lambda_i(\B{A} - \B{L}) = \lambda^\R{c}$ and $\mu_2(\B{A} - \B{L}) \le \mu_2(\B{A}) + \mu_2(-\B{L}) = \max_i \lambda_i(\B{A}) + 0 = \lambda^\R{u}$.

Next, we prove that $\beta \le \lambda^\R{c}$.
As $\sum_{i = 1}^{n_{\R{v}}} \lambda_i(\B{A} - \B{L}) = \R{tr}(\B{A} - \B{L})$,
\begin{equation*}
	\lambda^\R{c} \ge \frac{\R{tr}(\B{A} - \B{L})}{n_{\R{v}}} \ge
	\min_{\tilde{\B{L}}} \frac{\R{tr}(\B{A} - \tilde{\B{L}})}{n_{\R{v}}} = 
	\frac{1}{n_{\R{v}}} \left(\R{tr}(\B{A}) - \min_{\tilde{\B{L}}} \R{tr}(\tilde{\B{L}}) \right),
\end{equation*}
where $\R{tr}(\B{A}) = \sum_{i = 1}^{n_{\R{v}}} a_i$, and it is immediate to verify that $\min_{\tilde{\B{L}}} \R{tr}(-\tilde{\B{L}}) = - n_{\R{v}} (n_{\R{v}}-1)$, which happens when $- \tilde{L}_{ii} = - (n_{\R{v}}-1), \forall i \in \{1, \dots, n_\R{v}\}$, i.e., the graph is complete.
\end{proof}

\paragraph{Datasets' composition}%

All datesets contain 
$1$ complete graph,
$1$ path graph,
$1$ ring graph,
$n_{\R{v}}$ star graphs (each with a different center),
$1$ $2$-nearest neighbors graph,
a variable number of 
Erd\"os-Renyi graphs, 
small-world graphs,
scale-free graphs \cite{boccaletti2006complex},
and graphs with $e$ random edges, $\forall e \in \{n_{\R{e}}^\R{min}, \dots, n_{\R{e}}^\R{max}-1\} \setminus \{n_\R{e}^*\}$.
Disconnected graphs are discarded.
%

\paragraph{Datasets' coverage}

The number of connected labeled graphs with $10$ and $20$ vertices are 
$\approx 3.450 \cdot 10^{13}$ and
$\approx 1.569 \cdot 10^{57}$, respectively \cite{sloane2023number}.
We report the datasets' size and percent coverage of the decision spaces in Table \ref{tab:datasets}.


\paragraph{Node dynamics}

In $\C{D}_{\R{middle}}^\ell$ and $\C{D}_{\R{large}}^\ell$, 
$a_i = - n_\R{v} + (i-1)$;
in $\C{D}_{\R{small}}^\ell$, 
$a_i$ is chosen randomly in $[-20, -1]$.
In $\C{D}_{\R{middle}}^{\R{n}\ell}$ and $\C{D}_{\R{large}}^{\R{n}\ell}$, $a_i = 1 + 0.2 i$, and $\B{x}(t=0) = [-1 \ -2 \ -3 \ -4 \ -5 \ 2 \ 4 \ 6 \ 8 \ 10]\T$;
$e_{\R{thres}} = 0.01$, $t_{\R{max}} = 1$.

\paragraph{Neural networks}

Type of neural network: feedforward; 
optimizer: Adam; 
mini-batch size: $256$; 
learning rate: $0.01$, multiplied by $\gamma$ every $100$ episodes; 
activation functions: ``tanh''. 
For 
$\C{D}_{\R{middle}}^\ell$, 
$\C{D}_{\R{small}}^\ell$, 
$\C{D}_{\R{large}}^\ell$,
layers: 2 with 4 nodes each; 
epochs: $4000$; 
$\gamma = 0.95$. 
For
$\C{D}_{\R{middle}}^{\R{n}\ell}$, 
$\C{D}_{\R{large}}^{\R{n}\ell}$,
layers: 2 with 11 nodes each; 
epochs: $8000$; 
$\gamma = 0.975$.

\paragraph{Genetic algorithm}

Population size: 200; 
elite samples: $140$; 
crossover fraction 0.5; 
generations after which to stop if did not improve: 200; 
improvement tolerance on objective: $10^{-8}$; 
tolerance on constraints: $10^{-4}$.




\bibliographystyle{\bibliographyStyleName} 
\bibliography{\myBibliographyFile}

\end{document}